\documentclass[11pt]{article}
\usepackage[margin=25mm]{geometry}                % See geometry.pdf to learn the layout options. There are lots.
\usepackage{graphicx}
\usepackage{amssymb}
\usepackage{epstopdf}
\usepackage{amsmath}
\usepackage{amscd}
\usepackage{amsthm}
\usepackage{mathrsfs}
\usepackage{amscd}
\usepackage{color}
\usepackage[T1]{fontenc}
\usepackage{CJKutf8}
\usepackage{changes}
\usepackage[english]{babel}
\usepackage{tikz,tikz-3dplot}
\usepackage{bm}
\usepackage{xcolor}
\usepackage[normalem]{ulem}
\usepackage[colorinlistoftodos,prependcaption,textsize=tiny]{todonotes}

\newtheorem{thm}{Theorem}
\newtheorem{prop}[thm]{Proposition}

\title{MAD roots for large trees}
\author{David Bryant and Michael Charleston}
\date{\today}
\begin{document}
\maketitle

\begin{abstract}
The Minimal Ancestral Deviation (MAD) method is a recently introduced procedure for estimating the root of a phylogenetic tree, based only on the shape and branch lengths of the tree. 
The method is loosely derived from the midpoint rooting method, but, unlike its predecessor, makes use of all pairs of OTUs when positioning the root. 
In this note we establish properties of this method and then describe a fast and memory efficient algorithm. 
As a proof of principle, we use our algorithm to determine the MAD roots for simulated phylogenies with up to 100,000 OTUs. The calculations take a few minutes on a standard laptop.
\end{abstract}

\section{Introduction: the MAD method}

Phylogenetic inference methods usually reconstruct \emph{unrooted} trees, requiring an additional step to infer the position of the root.
Rooting is simple if the trees are \emph{ultrametric} or \emph{clock-like}; that is, if the root-to-tip distance is uniform or close to it.
In many or most situations however, trees are not ultrametric: in particular, when there are cases of \emph{heterotachy} --- change of evolutionary rate on some but not all branches, leading to apparent non-uniformity in the root-to-tip distances.

Minimal Ancestor Deviation (MAD) is a method designed to accommodate heterotachy in phylogenetic trees and rapidly obtain quality estimates of their roots. It is inspired by mid-point rooting, but rather than using a single pair of OTUs to locate the root, it uses all pairs. The full details of the method, and its derivation, can be found in \cite{MADpaper}. 
Here we outline the basic idea, modifying the presentation from \cite{MADpaper} but not altering the method.

\begin{figure}[ht]
\begin{center}
\includegraphics[width=0.7\textwidth]{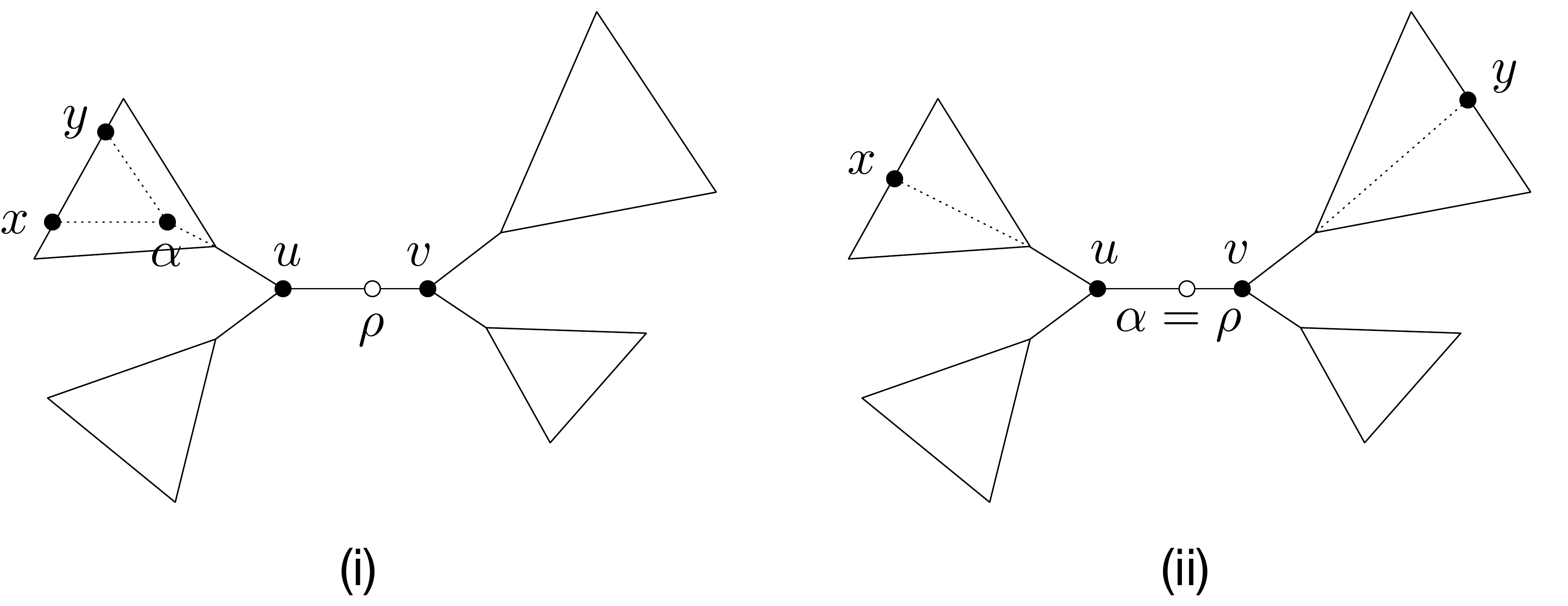}
\end{center}
\caption{\label{fig:MadPic} The relationship between pairs of OTUs and a putative root position. 
(i) The path from $x$ to $y$ does not pass through the putative root, and the least common ancestor $\alpha$ of $x$ and $y$ lies on the three-way intersection of the paths between $x$, $y$ and $\rho$. 
(ii) The path from $x$ to $y$ does pass through the putative root, and the least common ancestor equals $\rho$. }
\end{figure}

Consider the tree in Figure~\ref{fig:MadPic}. 
Let $d_{uv}$ denote the path length distance in the tree between any two nodes $u$ and $v$, where $u$ and $v$ could be OTUs (leaves), ancestral nodes, or positions along the branches.
Let $\rho$ denote the position of a putative root, where $\rho$ is located at a node or along a branch. 

Every pair of OTUs $x$ and $y$ has a unique least common ancestor $\alpha$ with respect to the putative root. 
There are two cases, depending on whether the $x$-$y$ path passes through $\rho$. 

Case 1: if the path from $x$ to $y$ does not pass through $\rho$, then $\alpha$ lies on the intersection of the paths from $x$ to $y$, $x$ to $\rho$, and $y$ to $\rho$ (Figure~\ref{fig:MadPic} (i)). 
If the tree was clock-like with root $\rho$, then we would have 
\[d_{x\alpha} = d_{y\alpha} = \frac{d_{xy}}{2}.\]
% We define the \emph{pairwise deviation} in a manner similar to \cite{MADpaper} as
% \[r_{xy;\alpha} = \left| \frac{2 d_{\alpha x} }{d_{xy} } - 1 \right| \]
% so that
% \[r_{xy;\alpha} = \left| \frac{2 d_{\alpha y} }{d_{xy} } - 1 \right| = \left| \frac{2 d_{\alpha x} - d_{\alpha y} }{d_{xy} }  \right|  =  \left| \frac{2 d_{ x \rho } - d_{ y \rho} }{d_{xy} }  \right| .\]
We define the \emph{pairwise deviation} in a manner similar to \cite{MADpaper} as
\[r_{xy;\alpha} = \left| \frac{2 d_{x\alpha} }{d_{xy} } - 1 \right| = \left| \frac{2d_{x\alpha} - d_{xy}}{d_{xy}}\right| = \left|\frac{2d_{x\alpha} - d_{x\alpha} - d_{y\alpha}}{d_{xy}}\right| = \left|\frac{d_{x\alpha} - d_{y\alpha}}{d_{xy}}\right|,\]
which is readily interpreted as the absolute proportional deviation of $\alpha$ from the half-way point between $x$ and $y$, and which is equal to
\[\left| \frac{d_{ x \rho } - d_{ y \rho} }{d_{xy} }  \right|\]
as the path from $\alpha$ to $\rho$ contributes to both $d_{x\rho}$ and $d_{y\rho}$.

Case 2: on the other hand, if as in Figure~\ref{fig:MadPic}(ii), the path from $x$ to $y$ does pass through $\rho$, then $\alpha = \rho$, (Figure~\ref{fig:MadPic} (ii)) and the corresponding pairwise deviation again becomes
\[r_{xy;\alpha} = \left| \frac{2 d_{y\alpha} }{d_{xy} } - 1 \right| =  \left| \frac{d_{ x \rho } - d_{ y \rho} }{d_{xy} } \right| .\]

Hence, irrespective of whether the path from $x$ to $y$ passes through $\rho$, we define the squared deviation
\[g_{xy}(\rho)  = \left(\frac{d_{x\rho} - d_{y\rho}}{d_{xy}}\right)^2\]
so that if $\alpha_{xy}$ denotes the least common ancestor of $x$ and $y$ then $g_{xy}(\rho) = g_{yx}(\rho) = r^2_{xy;\alpha_{xy}}$.
The overall deviation score for $\rho$ is now obtained by averaging this squared deviation over all pairs of OTUs, to obtain an {\em ancestor deviation score}
%\[r(\rho) = \left[ \sum_{x,y} (r_{xy;\alpha_{xy}})^2 \right]^{\frac{1}{2}} =  \left[ \sum_{x,y} g_{xy}(\rho) \right]^{\frac{1}{2}}.\]
 \[r(\rho) = \left[ \frac{2}{n(n-1)} \sum_{x,y} (r_{xy;\alpha_{xy}})^2 \right]^{\frac{1}{2}} =  \left[ \frac{2}{n(n-1)} \sum_{x,y} g_{xy}(\rho) \right]^{\frac{1}{2}}.\]

%\textbf{put averaging back into r(p); but minimise the G(rho)}

Minimizing $r(\rho)$ is clearly equivalent to minimizing
\[G(\rho) = \sum_{x,y} g_{xy}(\rho) =  \sum_{x,y}  \left(\frac{d_{x\rho} - d_{y\rho}}{d_{xy}}\right)^2. \]
We show below that $G(\rho)$ is a strictly convex function of the position $\rho$ on the tree, implying that $G(\rho)$ has a unique optimum. 
Our main result is that the optimal position for $\rho$ can be recovered in $O(n^2)$ time on an $n$ OTU tree, with $O(n)$ memory.   This is a significant and practical improvement over the $O(n^3)$ algorithm given by a direct implementation of the MAD formulas. 

\section{An efficient algorithm}

Let $T$ be an unrooted tree with $n$ OTUs. 
To begin with, we assume that $T$ is binary, though this can be relaxed (see below).  
We also assume that all branch lengths are non-negative.
For each pair $x,y$ of  OTUs in $T$ we define the function $f_{xy}(t)$ on the interval $[0,d_{xy}]$ by
\begin{align}
f_{xy}(t) &= \left(\frac{t - (d_{xy} - t)}{d_{xy}}\right)^2 = \left(\frac{2t - d_{xy}}{d_{xy}}\right)^{2} \nonumber \\
&= \frac{4t^{2}}{(d_{xy})^{2}} - \frac{4t}{d_{xy}} + 1, \label{eq:outer}
\end{align}
which is of the form $at^{2}+bt+c$, with the obvious values for $a$, $b$ and $c$.
% analogously with the definition of $g_{xy}(\rho)$.
% Note that $f_{xy}$ is a quadratic in $t$. 
% It is straightforward to show that if we define 
% \begin{align}
% a_{xy} &= \frac{4}{(d_{xy})^2} \nonumber\\
% b_{xy} & = -\frac{4}{d_{xy}} \nonumber \\
% c_{xy} & = 1. \label{eq:outer}\\
% \intertext{ then }
% f_{xy}(t) & = a_{xy} t^2 + b_{xy} t + c_{xy}.
% \end{align}
% This definition will allow us to minimise $G(\rho)$.

The following proposition characterises how the function $g_{xy}(\rho)$ changes as $\rho$ moves around the tree, and how that relates to $f_{xy}(t)$.

\begin{prop}
Consider a given pair of leaves $x$, $y$, and some arbitrary location $\rho$, which may or may not be on the $x$-$y$ path.   Let $\alpha$ be the location on the path from $x$ to $y$ that is closest to $\rho$. Then
\begin{equation}
g_{xy}(\rho) =  f_{xy}( d_{x\alpha}). \label{eq:gandf}
\end{equation}
\end{prop}
\begin{proof}
Let $t = d_{x\alpha}$. From the definitions of $g_{xy}(\rho)$ and $f_{xy}(t)$ we have
\begin{align}
g_{xy}(\rho) & = \left(\frac{d_{x\rho} - d_{y\rho}}{d_{xy}}\right)^2 \nonumber \\
& =  \left(\frac{d_{x \alpha} - d_{\alpha y}}{d_{xy}}\right)^2 \nonumber\\
& = \left(\frac{t - (d_{x y}-t)}{d_{xy}}\right)^2 \nonumber \\
& = f_{xy}(t). 
\end{align}
Intuitively, that part of the distance between $x$ (respectively $y$) and $\rho$ that does not lie on the $x$-$y$ path is either zero as $\rho$ lies on the path, or cancels out in the expression above.
\end{proof}

We make direct use of \eqref{eq:gandf} later. However first we demonstrate an important property of $G(\rho)$, and hence of $r(\rho)$. 

\begin{prop}
The function $G(\rho)$ is strictly convex on $T$. Hence there is a unique point $\rho$ minimizing $G(\rho)$ and any local optimum is a global optimum.
\end{prop}
\begin{proof}
For each $x,y$ the function $g_{xy}(\rho)$ coincides with with the strictly convex function $f_{xy}(d_{x\rho})$ on the path from $x$ to $y$. If we remove all branches on this path from $T$ then $g_{xy}(\rho)$ is constant on each of the components remaining. Hence $g_{xy}(\rho)$ is convex on $T$ and strictly convex on the path from $x$ to $y$. 

As $G(\rho) = \sum_{x,y} g_{xy}(\rho)$ is the sum of convex functions, it is itself convex. And as each pair of locations in $T$ is on the path connecting at least one pair of OTUs, $G$ is strictly convex.
\end{proof}

The next step is to define the functions $f_{uv}$ along each branch $uv$.
Let $uv$ be a branch of $T$.
Removing $uv$ partitions the set of OTUs into two parts: let $U$ be the set of OTUs closest to $u$ and let $V$ the be set of OTUs closest to $v$.  
We define the function $f_{uv}(t)$ on the interval $[0,d_{uv}]$ by
\begin{align}
f_{uv}(t) &= \sum_{x \in U} \sum_{y \in V} f_{xy}(t + d_{ux}).	
\end{align}
By \eqref{eq:gandf} we have that if $\rho$ is the point on the path from $u$ to $v$ that is distance $t$ from $u$ then 
\[f_{uv}(t) = \sum_{x \in U} \sum_{y \in V} g_{xy}(\rho).\]

The function $f_{uv}(t)$ is the sum of quadratic functions of $t$, so is itself a quadratic function of $t$. 
We let $a_{uv},b_{uv},c_{uv}$ denote the coefficients of the quadratic for each $u$, $v$, so that
\[f_{uv}(t) = a_{uv}t^2 + b_{uv}t + c_{uv}.\]
Note that $f_{uv}$ and $f_{vu}$ are not the same function.
%
%\begin{prop}
%Let $\{u,v\}$ be a branch of $T$. Then
%\[f_{uv}(t) = f_{vu}(d_{uv} - t).\]	
%
%
%\end{prop}
%\begin{proof}
%Let $U$ and $V$ the sets of OTUs as defined above. Then
%\begin{align*}
%	f_{uv}(t) & = \sum_{x \in U}\sum_{y \in V} f_{xy}(t+d_{ux}) \\
%	& = \sum_{x \in U}\sum_{y \in V} f_{yx}(d_{xy} -(t+d_{ux})) \\
%	& = \sum_{x \in U}\sum_{y \in V} f_{yx}(d_{vy} + (d_{uv} - t)) \\
%	& = f_{vu}(d_{uv}-t)
%\end{align*}
%Hence
%\begin{align*}
%f_{vu}(t) & = f_{uv}(d_{uv} - t) \\
%& = a_{uv} (d_{uv} - t)^2 + b_{uv} (d_{uv} - t) + c_{uv} \\
%& = a_{uv} t^2 - (2d_{uv} a_{uv} + b_{uv})t + a_{uv} d_{uv}^2 + b_{uv}d_{uv} + c_{uv}. 
%\end{align*}	
%\end{proof}

We will see below that once we have computed the coefficients $a_{uv}$ and $b_{uv}$ for each branch we can quickly determine the location $\rho$ which minimizes $G(\rho)$. Here we show how to compute these coefficients for all branches in $O(n^2)$ time. First note
\begin{align*}
f_{uv}(t) &= \sum_{x \in U} \sum_{y \in V} f_{xy}(t + d_{ux}) \\
& = \sum_{x \in U} \sum_{y \in V} a_{xy} t^2  + (b_{xy} + 2 a_{xy} d_{ux})t  + constant
\end{align*}
so that, by \eqref{eq:outer},
\begin{align}
a_{uv} & =  \sum_{x \in U} \sum_{y \in V} a_{xy} \nonumber \\
& =  \sum_{x \in U} \sum_{y \in V} \frac{4}{(d_{xy})^2} \label{eq:aexpand} \\
b_{uv} & =  \sum_{x \in U} \sum_{y \in V} (b_{xy} + 2a_{xy} d_{ux}) \nonumber \\
&=   \sum_{x \in U} \sum_{y \in V} \left( \frac{-4}{d_{xy}} + \frac{8d_{ux}}{(d_{xy})^2} \right) \label{eq:bexpand}.
\end{align}

There are two stages in the algorithm. In the first stage we compute $a_{uv}$ and $b_{uv}$ for each branch $uv$ such that $u$ is an OTU. 
 For the second stage, we compute $a_{uv}$ and $b_{uv}$ for all other branches in the tree. To do this we temporarily root the tree at an arbitrary OTU (the choice of OTU does not affect the final result). We then make use of the following recursion.

\begin{figure}
\centerline{\includegraphics[width=0.2\textwidth]{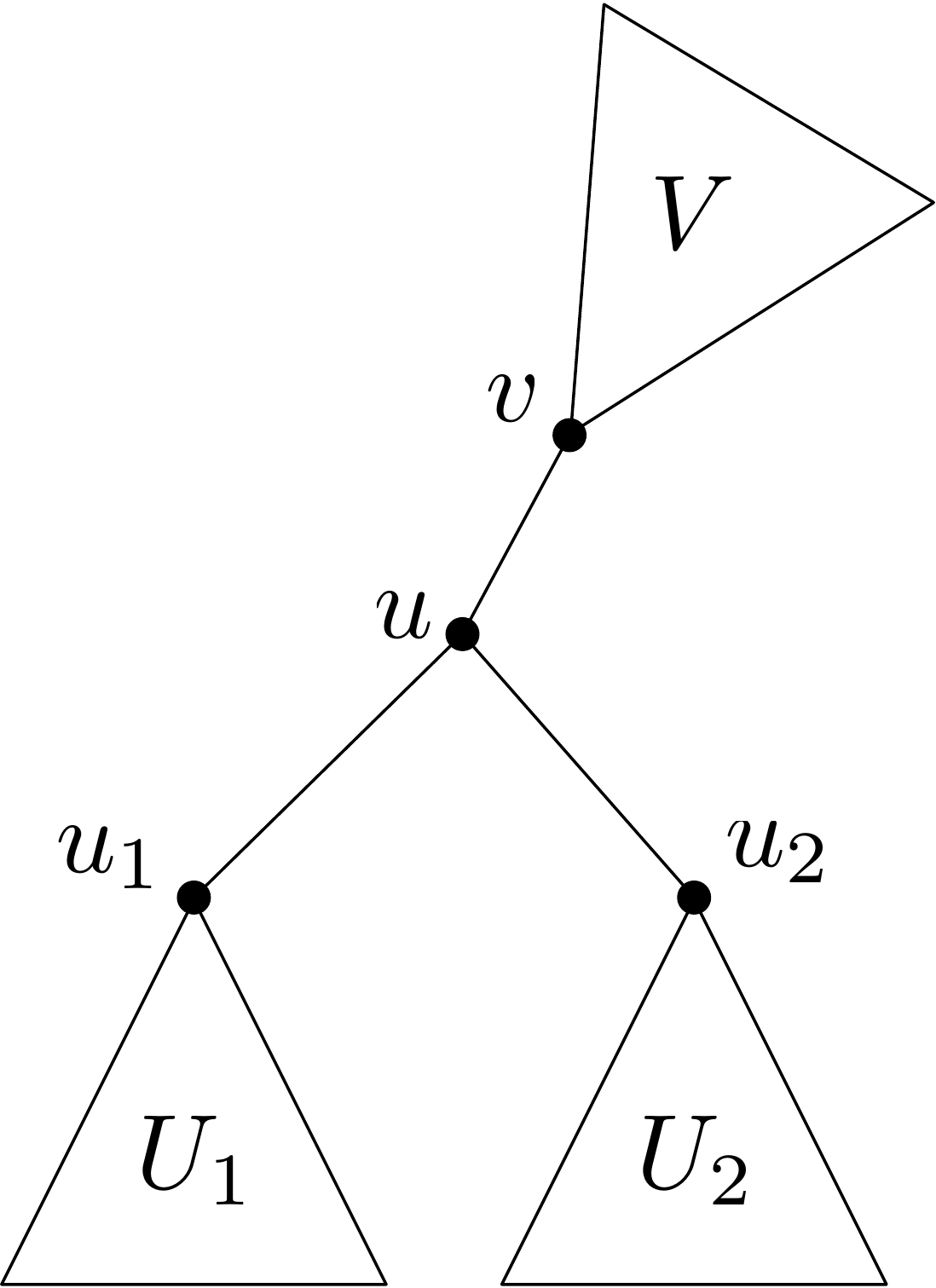}}
\caption{\small Sets of OTUs defined by the branch $\{u,v\}$. Let $U_1$ and $U_2$ denote the  sets of OTUs for the subtrees rooted at $u_1$ and $u_2$, and let $V$ be the set of remaining OTUs. \label{fig:tripleTree} }
\end{figure}

\begin{prop} \label{prop:abrecurse}
Let $u$ be an internal node in the tree, let $v$ be its parent and let $u_1,u_2$ be its children. 
Let $U_1,U_2$ be the OTUs for the subtrees rooted at $u_1$ and $u_2$ and let $V$ be the set of remaining OTUs, as illustrated in Figure~\ref{fig:tripleTree}. 
Then
\begin{align}
	a_{uv} & = a_{u_1 u} + a_{u_2 u}   - \sum_{x \in U_1} \sum_{y \in U_2} \frac{8}{(d_{xy})^2} \label{eq:arec} \\
	b_{uv} & =  b_{u_1u}  +  b_{u_2u}  + 2 d_{u_1u} a_{u_1 u}  + 2 d_{u_2u} a_{u_2 u}. \label{eq:brec}
\end{align}
\end{prop}
%\begin{prop} \label{prop:abrecurse}
%Let $uv$ be an internal branch of the tree and let $u_1,u_2$  be the two nodes adjacent to $u$ other than $v$. 
%Let $U_1,U_2$ be the OTUs closer to  $u_1,u_2$  than $v$, respectively, and let $V$ be the set of remaining leaves (Figure~\ref{fig:tripleTree}). 
%\added{\textbf{I want to say that these $U_{1}$, $U_{2}$ are the leaves of the subtrees rooted at $u_{1}$, $u_{2}$ and directed away from $u$, but it's still awkward\dots}}
%Then
%\begin{align}
%	a_{uv} & = a_{u_1 u} + a_{u_2 u}   - \sum_{x \in U_1} \sum_{y \in U_2} \frac{8}{(d_{xy})^2} \label{eq:arec} \\
%	b_{uv} & =  b_{u_1u}  +  b_{u_2u}  + 2 d_{u_1u} a_{u_1 u}  + 2 d_{u_2u} a_{u_2 u}. \label{eq:brec}
%\end{align}
%\end{prop}
\begin{proof}
\begin{align*}
	a_{uv} & = \sum_{x \in U_1} \sum_{z \in V} \frac{4}{(d_{xz})^2} +  \sum_{y \in U_2} \sum_{z \in V} \frac{4}{(d_{yz})^2}  \\
	& = a_{u_1u} + a_{u_2 u} - 2 \sum_{x \in U_1} \sum_{y \in U_2} \frac{4}{(d_{xy})^2}
	\intertext{ and }
	b_{u_1u}  +  b_{u_2u}  + 2 d_{u_1u} a_{u_1 u}  + 2 d_{u_2u} a_{u_2 u} \hspace{-4cm} & \\
	& = \sum_{x \in U_1} \sum_{y \in U_2}  \left(\frac{-4}{d_{xy}} + \frac{8d_{u_1x}}{(d_{xy})^2} \right) + \sum_{x \in U_1} \sum_{z \in V}  \left(\frac{-4}{d_{xz}} + \frac{8d_{u_1x}}{(d_{xz})^2} \right)  \\
	& \quad + \sum_{y \in U_2} \sum_{x \in U_1}  \left(\frac{-4}{d_{xy}} + \frac{8d_{u_2y}}{(d_{xy})^2} \right) + \sum_{y \in U_2} \sum_{z \in V}  \left(\frac{-4}{d_{yz}} + \frac{8d_{u_2y}}{(d_{yz})^2} \right)  \\
	& \quad + 2d_{u_1 u} \left( \sum_{x \in U_1} \sum_{y \in U_2}  \frac{4}{(d_{xy})^2}  +  \sum_{x \in U_1} \sum_{z \in V}  \frac{4}{(d_{xz})^2} \right) \\
	& \quad +  2d_{u_2 u} \left( \sum_{y \in U_2} \sum_{x \in U_1}  \frac{4}{(d_{xy})^2}  +  \sum_{y \in U_2} \sum_{z \in V}  \frac{4}{(d_{yz})^2} \right) \\
	& = \sum_{x \in U_1} \sum_{y \in U_2} \left( \frac{-4}{d_{xy}} + \frac{8d_{u_1x}}{(d_{xy})^2}  + \frac{-4}{d_{xy}} + \frac{8d_{u_2y}}{(d_{xy})^2} +  \frac{8d_{u_1 u} }{(d_{xy})^2}  +  \frac{8d_{u_2 u} }{(d_{xy})^2}  \right) \\
	& \quad + \sum_{x \in U_1} \sum_{z \in V}  \left(\frac{-4}{d_{xz}} + \frac{8d_{u_1x}}{(d_{xz})^2} +  \frac{8d_{u_1u}}{(d_{xz})^2} \right) \\
	& \quad + \sum_{y \in U_2} \sum_{z \in V}  \left(\frac{-4}{d_{yz}} + \frac{8d_{u_2y}}{(d_{yz})^2} +  \frac{8d_{u_2u}}{(d_{yz})^2} \right) \\
	& = \sum_{x \in U_1} \sum_{z \in V} \left( \frac{-8}{d_{xy}} + \frac{8(d_{xu_1} + d_{u_1u} + d_{uu_2} + d_{u_2y})}{(d_{xy})^2}\right) \\
	& \quad +  \sum_{x \in U_1} \sum_{z \in V}  \left(\frac{-4}{d_{xz}} + \frac{8d_{xu}}{(d_{xz})^2} \right) 
	+ \sum_{y \in U_2} \sum_{z \in V}  \left(\frac{-4}{d_{yz}} + \frac{8d_{yu}}{(d_{yz})^2}  \right) \\
	& = b_{uv}.
\end{align*}
\end{proof}

We note that the algorithm assumes that the tree $T$ is {\em binary}. To handle non-binary (multifurcating) trees we temporarily insert branches with zero length, in order to make them binary. This does not affect the values of $a_{uv}$ and $b_{uv}$ for the remaining branches. After the algorithm has completed, we remove the additional branches.

We now show how to quickly determine the location $\rho$ minimizing $G(\rho)$, using the coefficients $a_{uv}$ and $b_{uv}$ for each branch $uv$. 

\begin{prop} \label{prop:opt} 
\begin{enumerate}
	\item Let $uv$ be a branch of $T$. 
	Let $t = -\frac{b_{uv}}{2a_{uv}}$. 
	If $0 <t < d_{uv}$ then the location $\rho$ at distance $t$ along the branch from $u$ to $v$ is optimal. 
	\item Let $uv$ be a branch of $T$ such that $u$ is internal. 
	Let $u_1,u_2,\ldots,u_d$ be the nodes adjacent to $u$ other than $v$ (not assuming that $T$ is binary). 
	If $\frac{-b_{u_k u}}{2a_{u_ku}} \geq d_{u_k u}$  for all $k=1,2,\ldots,d$, and  $\frac{-b_{uv}}{2a_{uv}} \leq 0$, then the location $\rho = u$ is optimal.
	\item There is exactly one location in the tree which satisfies the first or second condition.
\end{enumerate}
\end{prop}

\begin{proof}
Let $uv$ be a branch in the tree, let $U$ be the set of OTUs closer to $u$ than $v$ and let $V$ be the complement of $U$. Let $\rho_t$ denote the location which is distance $t$ along the branch from $u$ to $v$, $0 <t<d_{uv}$. 
\begin{align*}
	\frac{d}{dt} G(\rho_t) & = \sum_{xy} \frac{d}{dt} g_{xy}(\rho_t) \\ 
	& = \sum_{x \in U} \sum_{y \in V} \frac{d}{dt} g_{xy}(\rho_t)  \\
	& = \sum_{x \in U} \sum_{y \in V} \frac{d}{dt}  f_{xy}(d_{xu} + t) & \mbox{ by  \eqref{eq:gandf} }\\
	& =  \frac{d}{dt}  f_{uv}(t) \\
	& = 2 a_{uv} t + b_{uv}.
\end{align*}
If $u$ is an OTU then $G(\rho_t)$ is strictly decreasing at $t=0$. Hence the optimal location for $G(\rho)$ is along a branch $uv$ or at an internal node $u$. The first case is characterized by a stationary point a solution of $\frac{d}{dt} f_{uv}(t) = 0$ for $0<t<d_{uv}$. The second case is characterized by $\frac{d}{dt} f_{u_1 u}(d_{u_1u}) \leq 0$, $\frac{d}{dt} f_{u_2 u}(d_{u_2u}) \leq 0$ and $\frac{d}{dt} f_{uv}(0) \geq 0$.
\end{proof}

These optimality conditions can be checked in constant time per edge. 

\begin{thm}
The MAD root for a tree with $n$ leaves can be determined in $O(n^2)$ time with $O(n)$ memory.
\end{thm}

\begin{proof}
There are three stages to the algorithm. In the first stage we compute the coefficients $a_{uv}$ and $b_{uv}$ for all  branches  connected to OTUs. This takes $O(n)$ time per external branch since we can compute the distance $d_{uy}$ from an OTU $u$ to every other OTU $y$ in linear time and then substitute these distances directly into \eqref{eq:aexpand} and \eqref{eq:bexpand}. Hence the first stage takes $O(n^2)$ time and $O(n)$ memory.

In the second stage we compute $a_{uv}$ and $b_{uv}$ for all internal branches in the tree. We  visit the internal nodes of the tree using a post-order traversal, noting that $a_{uv}$ and $b_{uv}$ have already been computed for all branches $uv$ where $u$ is an OTU.  When visiting node $u$, we first conduct a pre-order traversal of the subtrees $U_1$ and $U_2$ to compute and store the path lengths $d_{xu}$ for all $x \in U_1$ and $d_{yu}$ for all $y \in U_2$. We then evaluate \eqref{eq:arec} directly  in $O(|U_1||U_2|)$ time. We evaluate \eqref{eq:brec} in constant time.  Since $|U_1| |U_2|$ equals the number of pairs of OTUs with least common ancestor $u$, summing this over all internal nodes gives $\frac{n(n-1)}{2}$, the total number of pairs of OTUs. Hence the running time required to implement the recursions over all nodes in the tree is $O(n^2)$. The algorithm only requires $O(n)$ memory for this stage.

The final step, determining the actual optimum, takes only constant time per branch, or $O(n)$ time in total.
\end{proof}
\begin{figure}[ht]
	\begin{center}
	\includegraphics[width=0.9\textwidth]{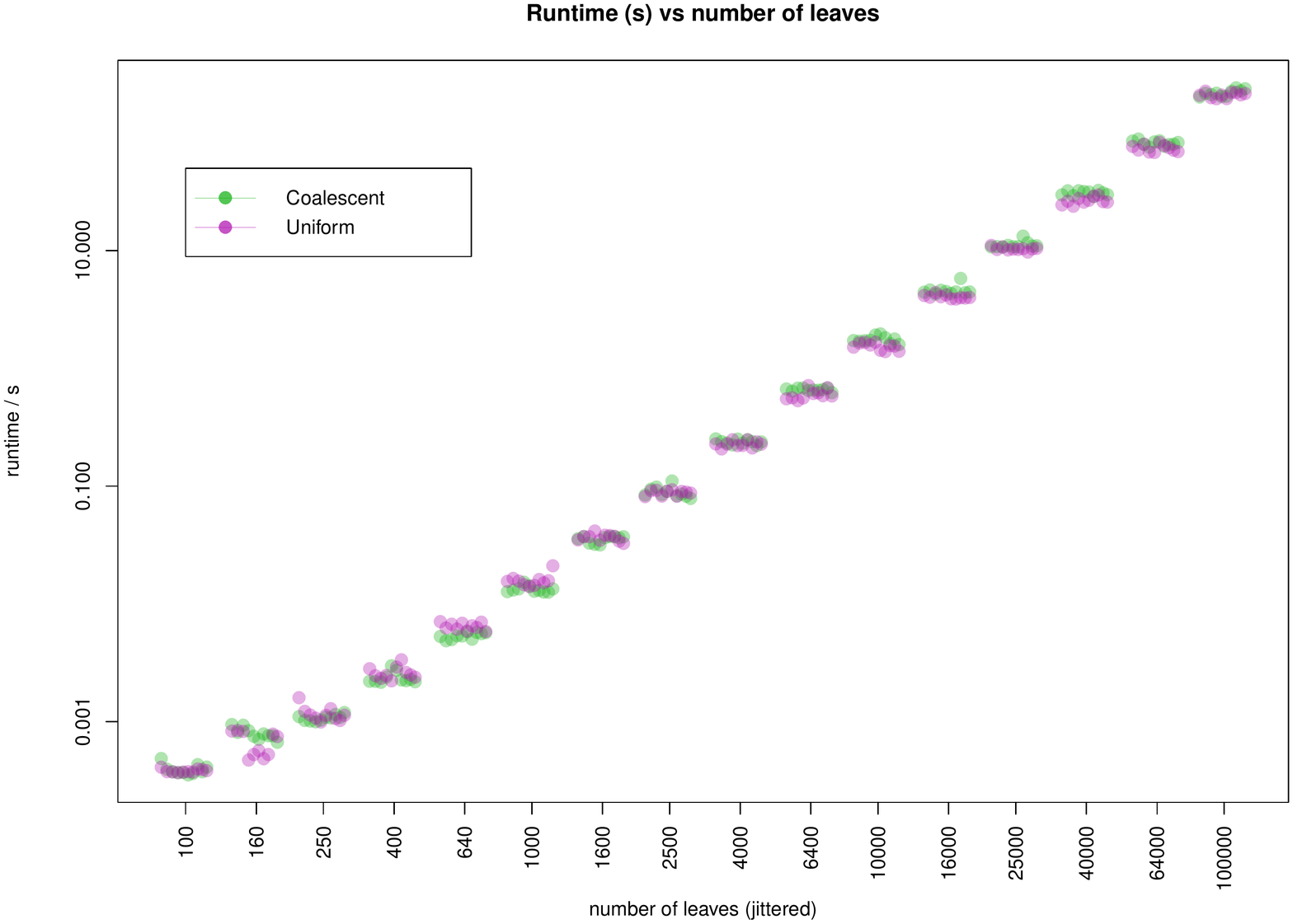}
	\end{center}
	\caption{Estimate time of our rooting algorithm as a function of $n$.}
	\label{fig:timing}
\end{figure}

\section{Experimental performance}

We have implemented our algorithm in open source C++, and code is available from either of the authors. The algorithm is fast. Figure~\ref{fig:timing} gives average running times for trees with 1000, 10000, and 100000 OTUs. For each replicate we simulated 10 trees drawn from the Yule distribution and 10 trees generated uniformly. 
Simulations were carried out on a  Mac PowerBook pro 3GHz Intel core i7 with 16Gb RAM.

We note that it only takes a few minutes to determine the MAD root for trees with 100,000 taxa. For smaller trees, the running time is negligible, meaning that the MAD root method could be applied to all trees in a large file with little computational cost, for example to incorporate branch length uncertainty explicitly into root location.

\bibliographystyle{plain}

\end{document}